\documentclass[12pt]{article}
\usepackage{graphicx}
\usepackage[utf8]{inputenc} 
\usepackage[T1]{fontenc}
\usepackage{url}
\usepackage{lmodern}
\usepackage{multirow}
\usepackage{savesym}
\usepackage{amsmath, amsfonts, amssymb, amsthm, amscd, amsxtra}
\usepackage{icomma}
\usepackage{numprint}
\usepackage{booktabs}
\usepackage[babel=true]{csquotes}
\usepackage{bbold}					
\usepackage{enumitem}       		
\usepackage{xcolor}
\usepackage{lipsum}					
\usepackage{pdflscape}				
\usepackage{xfrac}				  
\usepackage{bm}						
\usepackage{float}	
\usepackage[round]{natbib} 
\usepackage{tikz}
\usetikzlibrary{decorations.pathreplacing}
\usepackage{array}
\usepackage[main=english, francais]{babel}
\usepackage{caption}
\usepackage{hyperref}
\usepackage{verbatim}
\usepackage{titling}
\usepackage[top=3.5cm, bottom=2.5cm, left=3cm, right=3cm]{geometry}
\usepackage{threeparttable}		
\usepackage{subcaption}
\usepackage{tabularx}

\DeclareMathOperator*{\argmin}{arg\,min}

\newtheorem{proposition}{Proposition}

\newtheorem{corollary}{Corollary}

\graphicspath{{./Graphs/}}

\title{Flexible extreme thresholds through generalised Bayesian model averaging}

\author{Jessup, Sébastien\textsuperscript{1 *} \and Mailhot, Mélina\textsuperscript{1} \and Pigeon, Mathieu\textsuperscript{2}}
\date{\begin{flushleft}
        $^{1}${\fontsize{14}{11}\selectfont Department of mathematics, Concordia University, Montreal, Canada\\
        $^{2}$Département de mathématiques, UQAM, Montreal, Canada\\
		$^{*}$Corresponding author: sebastien.jessup@concordia.ca}
\end{flushleft}}

\begin{document}

\providecommand{\keywords}[1]{\textbf{Keywords } #1}
	
\maketitle

\begin{abstract}
    Insurance products frequently cover significant claims arising from a variety of sources. To model losses from these products accurately, actuarial models must account for high-severity claims. A widely used strategy is to apply a mixture model, fitting one distribution to losses below a given threshold and modeling excess losses using extreme value theory. However, selecting an appropriate threshold remains an open question with no universally agreed-upon solution. Bayesian Model Averaging (BMA) provides a promising alternative by enabling the simultaneous consideration of multiple thresholds. In this paper, we show that an error integration BMA algorithm can effectively detect heterogeneous optimal thresholds that adapt to predictive variables through the combination of mixture models. This method enhances model accuracy by capturing the full loss distribution and lessening sensitivity to threshold choice. We validate the proposed approach using simulation studies and an application to an automobile claims dataset from a Canadian insurer. As a special case, we also study the homogeneous setting, where a single optimal threshold is selected, and compare it to automatic selection algorithms based on goodness-of-fit tests applied to an actuarial dataset.
\end{abstract}

\begin{keywords}
	Bayesian Model Averaging; Extreme Values; Threshold Selection
\end{keywords}

\section{Introduction}\label{sec:Introduction}

Insurers are often exposed to significant losses from diverse sources, including injuries following accidents, crashes involving high-value vehicles, and natural catastrophes. In 2023, 23 of these events in Canada each resulted in over \$30 million in damages, contributing to nearly a quarter of the \$3.1 billion in total insured losses \citep{CatIQ2024}. As such, insurers require models that can capture both high-severity claims and the broader distribution of smaller claims to ensure a comprehensive risk assessment.

One approach to dealing with large losses is to use a branch of extreme value theory (EVT) called Peak-over-threshold. This approach examines values that exceed a specified level. For a sufficiently high limit, the excess values can be demonstrated to follow a generalised Pareto distribution (GPD), developed by \cite{Pickands1975}, with cumulative density function
\begin{equation}\label{eq:GPD}
	G(x) = \begin{cases}
		1-\left(1+\frac{\xi (x-u)}{\sigma}\right)^{-1/\xi} & \text{ for $\xi \ne 0$ and $x > u$}\\
		1-\exp\left(-\frac{x-u}{\sigma}\right) & \text{ for $\xi=0$ and $x > u$},
	\end{cases}
\end{equation}
with location, scale, and shape parameters $u, \sigma,$ and $\xi$. This approach can be used in a variety of contexts, such as extreme daily precipitation \citep{Thiombiano2017}, operational risk \citep{Chavez2006}, stock returns \citep{He2022}, and catastrophe risk \citep{Li2016}. The main challenge with this method is to identify a suitable threshold beyond which losses follow a GPD. Threshold selection, however, is an open problem with no universally accepted method. \cite{Caeiro2015} explain a few of the available methods based on a heuristic choice or minimisation of mean squared error.

What defines the optimal threshold varies depending on the method used, but it is generally accepted to represent the value which minimises the bias-variance tradeoff in the GPD tail of the distribution. For example, we can use EVT to determine this optimal value as the lowest point where the Mean Residual Life (MRL) plot becomes approximately linear \citep{Embrechts2013}.

Recently, in a homogeneous setting, automatic threshold selection methods have been proposed such as using L-moments \citep{Silva2020}, parameter stability \citep{Curceac2020}, goodness of fit \citep{Bader2018}, and other methods partly reviewed by \cite{Benito2023}, who find that different thresholds can yield similar market risk measures. These methods generally require establishing potential threshold values through a range of quantiles, from which a "best" threshold is chosen.

Although there is a significant body of literature devoted to identifying an optimal threshold, a question central to this article is whether it could be better to use model combination to simultaneously consider multiple threshold values. Such a combination could potentially be less affected by threshold misspecification than selecting a single best threshold. In fact, \cite{Northrop2017} recently used Bayesian Model Averaging (BMA), a method initially proposed by \cite{Raftery1997} that gained significant popularity in many scientific fields \citep{Fragoso2018}, to reduce sensitivity to threshold selection when studying the number of exceedances over a high threshold in ocean storms.

While \cite{Northrop2017} focus on the tail of the distribution, actuarial applications often require considering the entire range of potential values rather than only exceedances above a threshold. In such cases, models must capture the full distribution of the data, calling for an approach that models both the bulk and the tail. \cite{Laudage2019} proposed a mixture of a generalised linear model (GLM) and a GPD for a known threshold, which was expanded on by \cite{Ghaddab2023} in an excess-of-loss reinsurance context, where the threshold was estimated using MRL plots and Hill plot estimators. \cite{Li2023} proposed a three-part mixture model dividing low, medium and high claims according to a 20\%/60\%/20\% rule of thumb. While these approaches consider the full data, they are subject to threshold misspecification.

We propose using a tail-weighted version of a BMA algorithm, namely, the error integration algorithm proposed by \cite{Jessup2025} to combine mixture models with a single threshold, as suggested by \cite{Macdonald2011}. We will show that this allows for identifying the lowest threshold such that the distance between the fitted distribution and the observed tail data is minimised while improving the fit in the bulk of the distribution. We will also show that model combination is preferable to a single misspecified threshold. Furthermore, error integration allows for flexible thresholds depending on predictive variables.

As such, in this article, we propose a modification to a Bayesian model averaging algorithm to identify (flexible) optimal thresholds that allow the full distribution to be fitted, with an improved fit in the bulk of the distribution and reduced sensitivity to threshold selection. Section \ref{sec:Tail_weighted_BMA} establishes the theoretical results, Section \ref{sec:Data_application} illustrates these results using simulation studies and real datasets, and Section \ref{sec:Conclusion} concludes the article.

\section{Tail-weighted BMA for threshold selection}\label{sec:Tail_weighted_BMA}

To reduce the issue of threshold selection, we will focus on a mixture model proposed by \cite{Macdonald2011}. In a homogeneous setting, the authors suggested using a non-parametric distribution with kernels for the bulk of the data and a GPD for the tail beyond a threshold $u$. Replacing the kernel distribution by a parametric distribution, this has density
\begin{equation}\label{eq:mixture}
    f(y|\boldsymbol{\Lambda},u,\sigma,\xi)=\left\{
    \begin{array}{ll}
        (1-\phi_u) \times \frac{h(y|\boldsymbol{\Lambda})}{H(u|\boldsymbol{\Lambda})} & y \le u \\
        \phi_u \times g(y|u, \sigma, \xi) & y>u,
    \end{array}
    \right.
\end{equation}
where $h(y|\boldsymbol{\Lambda})$ is the bulk density with parameters $\boldsymbol{\Lambda}$, while $g(y|u,\sigma,\xi)$ is a GPD density with parameters $\sigma$ and $\xi$. Note that $\phi_u$ is simply the proportion of data above $u$ and is not a parameter as such. Furthermore, we can easily generalise (\ref{eq:mixture}) to take predictive variables into account To this end, consider a generalised mixture model with a fixed threshold such that
\begin{equation}\label{eq:general_mixture}
    f(y^{(k)}|\textbf{X}^{(k)},\boldsymbol{\Lambda}(\textbf{X}^{(k)}),u,\sigma(\textbf{X}^{(k)}),\xi(\textbf{X}^{(k)}))=\left\{
    \begin{array}{ll}
        (1-\phi_u) \times \frac{h(y|\textbf{X}^{(k)}, \boldsymbol{\Lambda}(\textbf{X}^{(k)}))}{H(u|\textbf{X}^{(k)}, \boldsymbol{\Lambda}(\textbf{X}^{(k)}))} & y \le u \\
        \phi_u \times g(y^{(k)}|u, \sigma(\textbf{X}^{(k)}), \xi(\textbf{X}^{(k)})) & y>u,
    \end{array}
    \right.
\end{equation}
where the prediction depends on the characteristics $\textbf{X}^{(k)}$ for the $k^\text{th}$ observation. Bulk distributions where parameters depend on predictive variables can be modeled by Generalised Additive Models for Location, Scale, and Shape (GAMLSS, see \cite{Rigby2005}). For the excess over the threshold, we can use Generalised additive extreme value models (evgam, see \cite{Youngman2022}).

\subsection{Mixture model combination}

We want to combine mixture models in order to reduce sensitivity to threshold selection. To this end, we use a BMA algorithm proposed by \cite{Jessup2025}. Fundamentally, weights are attributed to each of $M$ mixture models $\mathcal{M}_m$ to form a linear combination such that
\begin{equation}\label{eq:linearcombination}
	f(y) = \sum_{m=1}^{M}w_m f_m (y),
\end{equation}
where $\sum w_m=1$ and $f_m$ is the distribution under model $\mathcal{M}_m$. BMA sets the weights $w_m$ as the probability that each model is the true model given the observed data, or
\begin{equation}\label{eq:BMA}
	w_m = \Pr(\mathcal{M}_m|\mathcal{D}) = \frac{\Pr(\mathcal{D}|\mathcal{M}_m)\Pr(\mathcal{M}_m)}{\sum_{l=1}^M \Pr(\mathcal{D}|\mathcal{M}_l)\Pr(\mathcal{M}_l)},
\end{equation}
where $\Pr(\mathcal{D}|\mathcal{M}_m)$ is the likelihood of data $\mathcal{D}$ under $\mathcal{M}_m$, with prior probability $\Pr(\mathcal{M}_m)$.

\cite{Jessup2025} propose to obtain weights using uncertainty quantification, where modelling random error allows for numerical integration over this error, such that
\begin{equation}\label{eq:errorintegration}
    w_m = \Pr(\mathcal{M}_m|\mathcal{D}) \approx \frac{1}{S} \sum_{s=1}^S \frac{1}{|\mathcal{D}|} 
 \sum_{y^{(k)} \in \mathcal{D}} \frac{\Pr(y^{(k)}|\epsilon_s^{(k)},\mathcal{M}_m)\Pr(\mathcal{M}_m)}{\sum_{l=1}^M \Pr(y^{(k)}|\epsilon_s^{(k)},\mathcal{M}_l)\Pr(\mathcal{M}_l)},
\end{equation}
with $\mathcal{D}$ the observed data, $y^{(k)}$ the $k^\text{th}$ observation, $\epsilon_s^{(k)}$ the $s^\text{th}$ realisation of random error $\epsilon^{(k)}$, and $S$ the number of simulations. In the context of mixture models, random error can be assumed to follow a skew-normal for the bulk of the distribution, and a GEV distribution for the extreme value tail.

Furthermore, although these values account for only a small fraction of events, they represent the scenarios that require the most attention. As such, we want to ensure that this portion of the distribution is well-modelled. By taking a weighted mean over the $K$ different observations such that
\begin{equation}\label{eq:weightederrorintegration}
    w^*_m = \Pr\hspace{0.1em}^* (\mathcal{M}_m|\mathcal{D}) \approx \frac{1}{S} \sum_{s=1}^S \sum_{y^{(k)} \in \mathcal{D}} \frac{y^{(k)}}{\sum_{i=1}^K y^{(i)}} \frac{\Pr(y^{(k)}|\epsilon_s^{(k)},\mathcal{M}_m)\Pr(\mathcal{M}_m)}{\sum_{l=1}^M \Pr(y^{(k)}|\epsilon_s^{(k)},\mathcal{M}_l)\Pr(\mathcal{M}_l)},
\end{equation}
the weight to the $m^\text{th}$ model will depend more on the tail of the distribution. Note that the choice of weights is motivated by the intuition of losses receiving a weight corresponding to their percentage of total loss, and that other reweighting methods are possible.

\cite{Jessup2025} further suppose that the vectors of $\mathbf{w}_m^{(k)}$ follow a Dirichlet distribution, which allows for a regression depending on predictive variables such that
\begin{equation}\label{eq:Dirichletcombination}
    f(y^{(k)}) = \sum_{m=1}^M w^{(k)}_m f_m(y^{(k)}).
\end{equation}
This is particularly promising in the context of extreme value analysis, where comparing (\ref{eq:errorintegration}) and (\ref{eq:weightederrorintegration}) allows for Proposition \ref{prop:Dirichletthreshold}.

\begin{proposition}\label{prop:Dirichletthreshold}
    Let $M$ different models $\mathcal{M}_m$, $m \in \{1, \ldots, M\}$ be generalised mixture models as defined by equation \ref{eq:general_mixture} $\forall ~m$, with each model having a different predetermined threshold $u_m$ such that the models cover a wide range of possible threshold values from low to high values, such that the optimal threshold is in this range. Let $w_m^{(k)}$ and $w_m^{*(k)}$ be the weights to the $m^\text{th}$ model under respectively the mean and weighted-mean Dirichlet regression for the $k^\text{th}$ observation. Further let $u^{(k)}$ be the optimal threshold that minimises the GPD bias-variance tradeoff for the $k^\text{th}$ observation. Then, $w_m^{*(k)} \ge w_m^{(k)}$ if $u_m \ge u^{(k)}$, and $w_m^{*(k)} \le w_m^{(k)}$ if $u_m \le u^{(k)}$.
\end{proposition}

\begin{proof}
    Consider a combination of two mixture models such that $u_1 < u^{(k)}$ and $u_2 \ge u^{(k)}$, for $u^{(k)}$ the optimal threshold minimising the GPD bias-variance tradeoff for the $k^\text{th}$ observation $y^{(k)}$.
    
    From the Pickands-Balkema-de-Haans theorem \citep{Coles2001}, $\mathcal{M}_2$ must have lower bias than $\mathcal{M}_1$ when $y^{(k)} > u^{(k)}$. This implies that on average, 
    \begin{equation*}
        |y^{(k)}-\hat{y}^{(k)}_1| > |y^{(k)}-\hat{y}^{(k)}_2|,
    \end{equation*}
    such that from the combination hypotheses of error integration, on average for $y^{(k)} > u^{(k)}$,
    \begin{equation*}
        \Pr(y^{(k)}|\epsilon^{(k)}_s,\mathcal{M}_1) < \Pr(y^{(k)}|\epsilon^{(k)}_s,\mathcal{M}_2).
    \end{equation*}
    Furthermore, for these values,
    \begin{equation*}
        \frac{y^{(k)}}{\sum_{i=1}^K y^{(i)}} > \frac{1}{|\mathcal{D}|}.
    \end{equation*}
    Comparing equations (\ref{eq:errorintegration}) and (\ref{eq:weightederrorintegration}), it then follows that $w^{*(k)}_2$ must increase relative to $w^{(k)}_2$. Finally, since $\sum w^{*(k)}_m=1$, $w^{*(k)}_1$ must decrease relative to $w^{(k)}_1$.

    This can be generalised to $M$ models by induction. For each additional model, if $u_m \ge u^{(k)}$ ($u_m < u^{(k)}$), then from the same argument $w^{*(k)}_m \ge w^{(k)}_m$ ($w^{*(k)}_m \le w^{(k)}_m$). It directly follows that the result holds for all $M$ models.
\end{proof}

Proposition \ref{prop:Dirichletthreshold} allows for identifying flexible thresholds based on covariates while providing a distribution for the full range of data. This represents an improvement over methods like quantile regression, where, for instance, \cite{Youngman2019} outlines a two-step procedure: using a quantile regression to identify thresholds and then fitting a GAM version of the GPD to the excesses over those thresholds. Such methods, while efficient, do not model the full distribution.

In the absence of predictive variables, residuals cannot be directly computed for each observation. However, bootstrapping provides a way to evaluate residuals around each quantile, allowing us to adapt equations (\ref{eq:errorintegration}) and (\ref{eq:weightederrorintegration}) by considering quantiles instead of individual observations. Through bootstrapping, we can estimate both variance and skewness at each quantile, which allows us to use the error integration framework with quantiles instead of observations. This approach relaxes the assumption of asymptotic normality in quantile estimates \citep{Vandervaart2000}, addressing potential skewness that arises from limited data availability.

To distinguish between different forms of error integration, we define the homogeneous quantile-based approach as hom-EI and the heterogeneous version as het-EI. Note that hom-EI is very similar to Generalised Uncertainty Likelihood Estimation (GLUE, e.g. \cite{Jessup2023}, \cite{Zhu2013}), except that we allow for a skew-normal distribution around the quantiles rather than a normal distribution. Further details on the skew-normal distribution are provided in Appendix \hyperref[sec:skewnorm]{I}, while a skew-normal version of the GLUE algorithm is in Appendix \hyperref[sec:SNGLUE]{II}. Importantly, hom-EI enables Corollary \ref{cor:GLUE}.

\begin{corollary}\label{cor:GLUE}
    Let $M$ different models $\mathcal{M}_m$, $m \in \{1,\ldots, M\}$, be mixture models as defined by equation (\ref{eq:mixture}) $\forall ~m$, with each model having a different predetermined threshold $u_m$ s.t. the models cover a wide range of possible thresholds from low to high values, such that the optimal threshold is in this range. Let $w_m$ and $w^*_m$ be the weights to the $m^\text{th}$ model under respectively hom-EI and weighted hom-EI. Further let $u$ be the optimal threshold that minimises the GPD bias-variance tradeoff. Then, $w^*_m \ge w_m$ if $u_m \ge u$, and $w^*_m \le w_m$ if $u_m \le u$.
\end{corollary}

\begin{proof}
    The proof is very similar to the proof of Proposition \ref{prop:Dirichletthreshold}. Instead of considering the $k^{th}$ observation, we consider the $q^\text{th}$ quantile. The same reasoning applies, where lower bias implies higher average likelihood.
\end{proof}

In particular, Corollary \ref{cor:GLUE} allows for comparison with automatic threshold selection algorithms. Note that the proofs of Proposition \ref{prop:Dirichletthreshold} and Corollary \ref{cor:GLUE} both assume that the optimal threshold lies within the wide range of considered thresholds. It is not intuitively clear what would happen when all models have thresholds above or below the best threshold, highlighting the importance of considering low to high values. Further, due to data uncertainty, thresholds should be sufficiently spaced out to have significant difference between the models.

Proposition \ref{prop:Dirichletthreshold} and Corollary \ref{cor:GLUE} allow us to identify optimal (flexible) thresholds minimising the bias-variance tradeoff for the GPD tail. While this is certainly desirable, a natural question in the context of model combination is whether combined models can outperform the model fitted with the right threshold. When fitting the mixture model in equation (\ref{eq:mixture}), the threshold automatically implies truncation of the left part of the data. This truncation in turn means that the parameters obtained through MLE for the bulk of the data will be biased. Note that while this bias can be reduced by using censored MLE (see for example \cite{Zeng2007}), since the truncated observations affect the evaluation of parameters in a finite data setting, some bias will remain. Model combination can reduce this bias by considering multiple parameters simultaneously. In theory, a combination should more closely approximate the true distribution than a model fitted with the optimal threshold. As such, Section \ref{sec:Data_application} illustrates both the threshold identification and accuracy of the combinations compared to single thresholds.

\section{Data application}\label{sec:Data_application}

\subsection{Simulation study}

To illustrate threshold selection using Proposition \ref{prop:Dirichletthreshold} and Corollary \ref{cor:GLUE}, we first use simulated databases for which we know the true optimal threshold(s). 

\subsubsection{Homogeneous example}

We start with a homogeneous simulated dataset built with Equation (\ref{eq:mixture}), where the bulk of the data follows a lognormal distribution with parameters $\mu=5, \sigma=1$, and the tail, representing 5\% of the data, follows a GPD with parameters $u = 700, \sigma_u = 600$, and $\xi = 0.2$. Figure \ref{fig:MRL_simulation} illustrates the MRL plot for this simulated data. From this graph, plausible points where we see changes in slope followed by an approximately linear portion are slightly under 600, a clear break at 700, then at 1000.

\begin{figure}[h]
    \centering
    \includegraphics[width=0.7\linewidth]{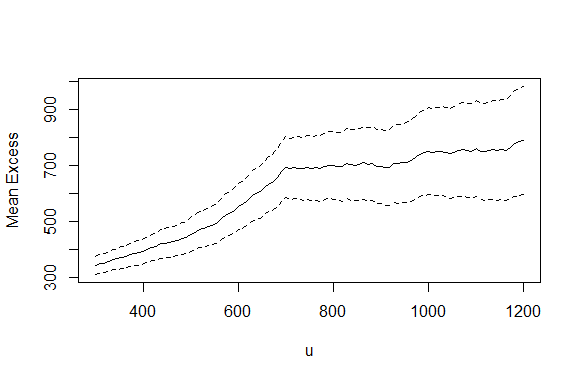}
    \caption{MRL for a mixture of a LN($\mu=5,\sigma=1$) and GPD($u=700,\sigma_u=600,\xi=0.2$)}
    \label{fig:MRL_simulation}
\end{figure}

To use Corollary \ref{cor:GLUE}, we need a wide range of values above and below the threshold. Consider three different ranges: 200 to 800, 400 to 1000, and 600 to 1200. In each case, the optimal threshold of 700 lies within the upper, middle, or lower portion of the specified range. The resulting model weights, shown in Figure \ref{fig:Simulated_weights}, indicate that the weight reversal effectively identifies the correct threshold: when 700 is located in the middle or lower part of the range, a weight reversal occurs precisely at this threshold. For the range where 700 lies in the upper values, the weight reversal is close to, though not exactly at, 700.

\begin{figure}[h]
\centering
\begin{minipage}{0.49\linewidth}
    \includegraphics[width=\linewidth]{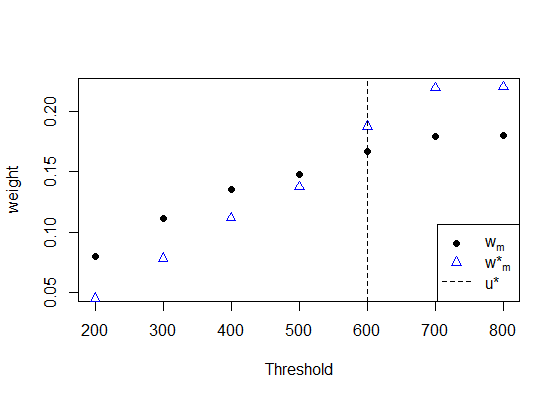}
\end{minipage}
\begin{minipage}{0.49\linewidth}
    \includegraphics[width=\linewidth]{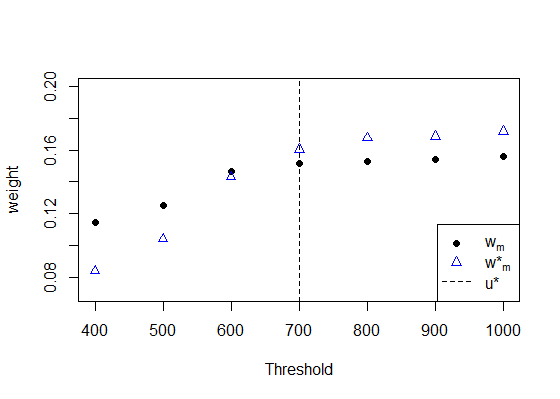}
\end{minipage}\\
\begin{minipage}{0.49\linewidth}
\centering
    \includegraphics[width=\linewidth]{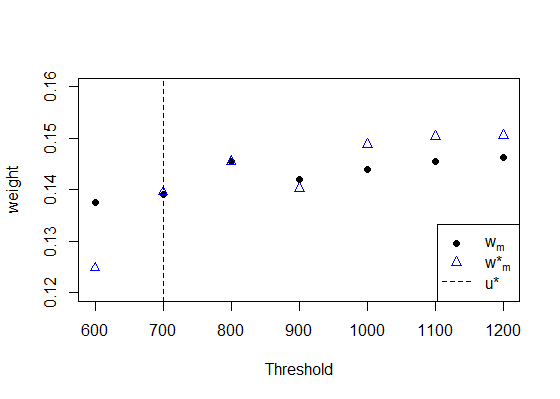}
\end{minipage}  
    \caption{Model weights $w_m$ and $w^*_m$ for thresholds ranging from 200 to 800 (top left), 400 to 1000 (top right), and 600 to 1200 (bottom) with the identified optimal threshold $u^*$}
    \label{fig:Simulated_weights}
\end{figure}

In order to support the argument that a combination can outperform a fitted model with the right threshold, consider the Hellinger distance \citep{Beran1977} and the Kullback-Leibler (KL) divergence \citep{VanErven2014}, defined in Appendix \hyperref[sec:measures]{III}. When adjusting the model with a sufficiently high threshold, the parameters for the tail of the distribution are accurate, but the estimated parameters for the lognormal bulk of the data are systematically biased in a finite dataset. Indeed, right truncation leads to the location parameter being underestimated, where
\begin{equation*}
	\hat{\mu} = \frac{\sum_{i=1}^n \ln(x_i) I(x_i < u)}{\sum_{i=1}^n I(x_i<u)} < \mu.
\end{equation*}
This should in turn lead to the distance and divergence between this model and the data being higher than with a combination of multiple thresholds.

Table \ref{tab:simulation_homogeneous_distance} shows the Hellinger distance and KL divergence comparing the empirical distribution of the test set to the distributions obtained using hom-EI and weighted hom-EI (w-hom-EI) combinations, as well as the reference value of 700, correctly identified using Corollary \ref{cor:GLUE}. In the bulk of the distribution, w-hom-EI slightly outperforms hom-EI, and both methods outperform the mixture model with the right threshold. In the tail, the model based on the identified value achieves the best performance. This is expected, as the chosen value matches the simulated true threshold, ensuring the model is unbiased for extreme values.

\begin{table}[h]
    \centering
    \begin{tabular}{ccccc}
    \toprule
    & \multicolumn{2}{c}{Bulk ($x \le 700$)} & \multicolumn{2}{c}{Tail ($x > 700$)}\\
    Combination & Hellinger ($*10^{-7}$) & KL ($*10^{-3}$) & Hellinger ($*10^{-5}$)& KL ($*10^{-2}$)\\
    \midrule
    hom-EI & 1.38 & 2.70 & 1.044 & 1.0091\\
    w-hom-EI & 1.35 & 2.65 & 1.041 & 1.0055\\
    700 & 1.95 & 3.84 & 0.010 & 0.0095\\
    \bottomrule
    \end{tabular}
    \caption{Hellinger distance and Kullback-Leibler divergence for the bulk under 700 and the tail beyond 700 by combination method}
    \label{tab:simulation_homogeneous_distance}
\end{table}

Given that the database is simulated, we can further confirm these results by simulating 10,000 multiple databases to obtain the average threshold as well as the proportion where hom-EI and w-hom-EI outperformed the single threshold. We find a mean of 697 and standard deviation of 17, with the proportions where the combinations outperformed the single threshold in Table \ref{tab:proportions_simulation_homogeneous}. These proportions lend weight to the intuition that combinations reduce bias for the bulk of the data, but cannot perform better than the identified threshold for the tail, given that this threshold optimises the bias-variance tradeoff. We also see that w-hom-EI tends to perform better than hom-EI.

\begin{table}[h]
    \centering
    \begin{tabular}{ccccc}
     \toprule
    & \multicolumn{2}{c}{Bulk ($x \le 700$)} & \multicolumn{2}{c}{Tail ($x > 700$)}\\
    Combination & Hellinger & KL & Hellinger & KL\\
    \midrule
    hom-EI & 89.76\% & 92.56\% & 3.69\% & 3.68\%\\
    w-hom-EI & 99.99\% & 99.99\% & 3.79\% & 3.80\%\\
    \bottomrule
    \end{tabular}
    \caption{Proportion of simulations where the combination method outperforms the identified threshold for a homogeneous database}
    \label{tab:proportions_simulation_homogeneous}
\end{table}

\subsubsection{Heterogeneous example}

We now turn to a heterogeneous example to illustrate Proposition \ref{prop:Dirichletthreshold}. Suppose that there are two types of claims, each occurring 50\% of the time. The first type is the same as the previous example, where the first 95\% follows a lognormal with $\mu=5, \sigma=1$ and the remaining 5\% tail follows a GPD with $u=700, \sigma_u=600$, and $\xi=0.2$. The second type is also a mixture with the first 95\% following a lognormal and the 5\% tail following a GPD, but with parameters $\mu=6, \sigma=0.5, u=1000, \sigma_u=500$, and $\xi=0.1$. Further suppose that initial claim appraisal allows for identifying whether the claims will be small, medium, or large, corresponding to the lowest 20\%, the next 60\%, and the highest 20\%.

Figure \ref{fig:MRL_simulation2} shows the MRL plot for the overall data. There is a clear change at 1000 followed by a straight line, but the known threshold of 700 is much less obvious. To account for all potential points where the slope changes, we select thresholds as 400 to 1600 by jumps of 200 to obtain the weights in Figure \ref{fig:Simulated_weights_EI}. The first weight reversal happens near the correct threshold of 700 and at 1000 for the first and second types of loss respectively. Note however that the result is not as obvious as the homogeneous case. This may be due to the limited predictive variables introducing substantial uncertainty, which reduces the effectiveness of the error integration approach.

\begin{figure}[H]
    \centering
    \includegraphics[width=0.7\linewidth]{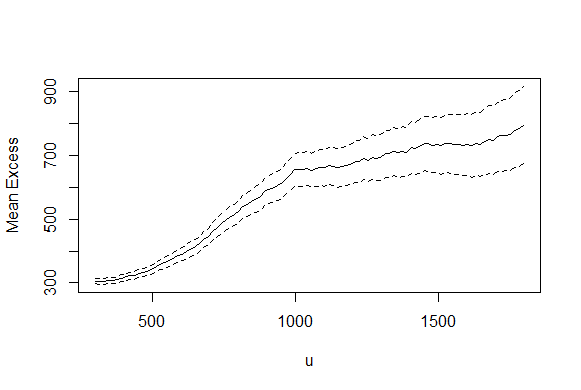}
    \caption{MRL for the second simulated database with two types of loss, where type 1 is a mixture of a LN($\mu=5,\sigma=1$) and GPD($u=700,\sigma_u=600,\xi=0.2$) and Type 2 is a mixture of a LN($\mu=6,\sigma=0.5$) and a GPD($u=1000,\sigma_u=500,\xi=0.1$)}
    \label{fig:MRL_simulation2}
\end{figure}

\begin{figure}[H]
    \centering
    \begin{minipage}{0.49\linewidth}
        \includegraphics[width=\linewidth]{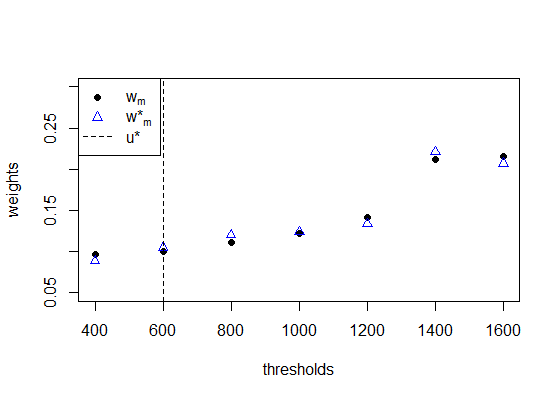}
    \end{minipage}
    \begin{minipage}{0.49\linewidth}
        \includegraphics[width=\linewidth]{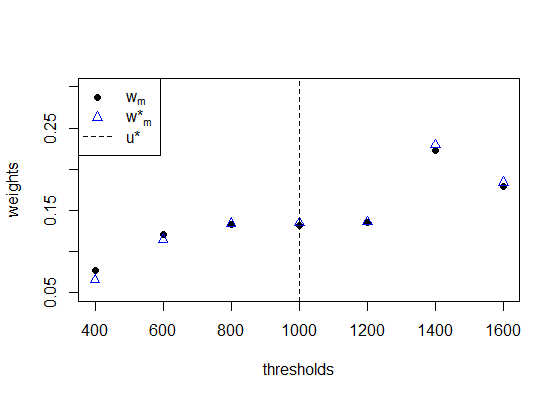}
    \end{minipage}
    \caption{Model weights $w_m$ and $w_m^*$ for Type 1 (left) and Type 2 (right) with identified thresholds $u^*$}
    \label{fig:Simulated_weights_EI}
\end{figure}

We can again compare distance and divergence measures for the (weighted) error integration combination against the threshold identified using Proposition \ref{prop:Dirichletthreshold}. The results, presented in Table \ref{tab:simulation_heterogeneous_distance}, show that model combination outperforms the single value model in the bulk of the distribution for type 1 but not for type 2. This difference may be attributed to the limited predictive variables, where variations in cutoff points and simulation parameters lead to better performance of the het-EI algorithm for type 1 than for type 2. This makes sense, as the identified thresholds are near the true values used to simulate data in this case, such that the tail should have very low bias with the single threshold.

\begin{table}[h]
    \centering
    \resizebox{\linewidth}{!}{
    \begin{tabular}{cccccc}
    \toprule
    & & \multicolumn{2}{c}{Bulk ($x \le u^{(k)}$)} & \multicolumn{2}{c}{Tail ($x > u^{(k)}$)}\\
    & Threshold & Hellinger ($*10^{-7}$) & KL ($*10^{-3}$) & Hellinger ($*10^{-4}$)& KL\\
    \midrule
    \multirow{3}{*}{\centering Type 1} & het-EI & 1.32 & 2.51 & 2.90 & 0.39\\
    & w-het-EI & 1.33 & 2.55 & 2.96 & 0.40\\
    & 600 & 3.19 & 6.22 & 0.003 & 0.04\\
    \midrule
    \multirow{3}{*}{\centering Type 2} & het-EI & 4.57 & 8.78 & 9.17 & 1.61\\
    & w-het-EI & 3.60 & 6.88 & 9.01 & 1.59\\
    & 1000 & 1.12 & 2.18 & 0.003 & 0.0003\\
    \bottomrule
    \end{tabular}
    }
    \caption{Hellinger distance and Kullback-Leibler divergence for the bulk under $u^{(k)}$ and tail above $u^{(k)}$ by type of loss $k$, where $u^{(1)*}=600$ and $u^{(2)*}=1000$}
    \label{tab:simulation_heterogeneous_distance}
\end{table}

Similarly to the homogeneous case, we can use the artificial nature of the data to run 10,000 simulations. For Type 1, we obtain a mean of 674 with standard deviation of 100, while for Type 2 we obtain a mean of 978 and standard deviation of 182. In both cases, the mean is close to the true value. We can further obtain the proportions in Table \ref{tab:proportions_simulation_heterogeneous}. These proportions further suggest that while w-het-EI outperforms het-EI, the limited predictive variables make the algorithm less optimal for type 2 losses.

\begin{table}[h]
    \centering
    \begin{tabular}{cccccc}
    \toprule
    & & \multicolumn{2}{c}{Bulk ($x \le u^{(k)}$)} & \multicolumn{2}{c}{Tail ($x > u^{(k)}$)}\\
    & Threshold & Hellinger & KL & Hellinger & KL\\
    \midrule
    \multirow{2}{*}{\centering Type 1} & het-EI & 99.88\% & 99.88\% & 0\% & 0.04\%\\
    & w-het-EI & 99.96\% & 99.96\% & 0\% & 0.05\%\\
    \midrule
    \multirow{2}{*}{\centering Type 2} & het-EI & 23.28\% & 24.72\% & 0\% & 0.32\%\\
    & w-het-EI & 32.28\% & 32.64\% & 0\% & 0.33\%\\
    \bottomrule
    \end{tabular}
    \caption{Proportion of simulations where the combination method outperforms the identified threshold for a heterogeneous database}
    \label{tab:proportions_simulation_heterogeneous}
\end{table}

\subsection{Actuarial data}

The simulation studies have demonstrated how Proposition \ref{prop:Dirichletthreshold} and Corollary \ref{cor:GLUE} allow for threshold identification in both homogeneous and heterogeneous settings. We can now look at real actuarial data to obtain similar results. For real data, however, the "true" optimal threshold is unknown, and we aim to compare our identified values with those derived from automatic selection methods available in the literature.

\subsubsection{Homogeneous database}

We first look at a well-known reinsurance dataset, namely the Danish data \citep{Mcneil1997}. There are 2,167 losses between 1 and 263 million Danish krones, expressed in millions. We randomly split the data into a training and testing set, where both sets have approximately the same size. We then fit mixture models where we suppose that the bulk of the data follows a lognormal distribution and the tail follows a GPD. \cite{Embrechts2013} suggested that a threshold of 10 or 18 (million) is appropriate for this dataset based on the MRL plot shown in Figure \ref{fig:DanishMRL}. We can see that around the suggested points, the mean excess is approximately linear from 10 to 18, then from 18 to 30 with a different slope. This method of selecting a threshold is highly subjective, but gives a reasonable idea of where the threshold might be.
\begin{figure}[H]
    \centering
    \includegraphics[width=0.7\linewidth]{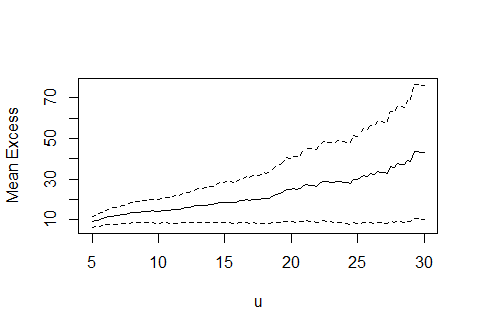}
    \caption{Danish MRL plot}
    \label{fig:DanishMRL}
\end{figure}

With the weights illustrated in Figure \ref{fig:Danishweights}, we can apply Corollary \ref{cor:GLUE} to obtain the same optimal threshold of 10 (million).
\begin{figure}[H]
    \centering
    \includegraphics[width=0.7\linewidth]{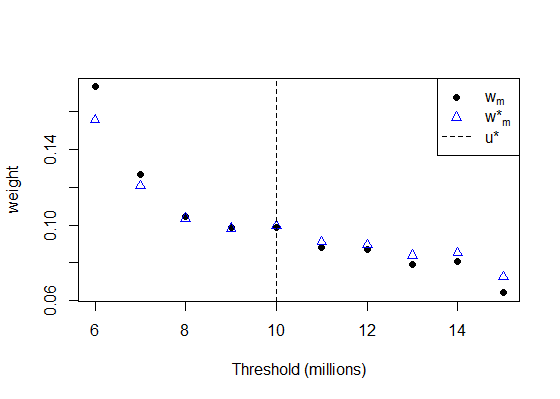}
    \caption{Danish model weights for thresholds from 6 to 15}
    \label{fig:Danishweights}
\end{figure}

Table \ref{tab:thresholdsDanish} further highlights a range of potential thresholds using an Anderson-Darling ForwardStop (FS) algorithm \citep{Bader2018} as well as two methods available in the "\hyperlink{https://cran.r-project.org/web/packages/tea/tea.pdf}{tea}" package in R \citep{Ossberger2020} along with their corresponding quantiles. The Anderson-Darling ForwardStop (FS) algorithm involves fitting a GPD to data above a set of increasing thresholds $\{u_1, \ldots, u_n\}$, then using the Anderson-Darling p-value to identify the first threshold $\hat{k}_{FS}$ that exceeds a certain significance level $\alpha$, where
\begin{equation}\label{eq:FS}
    \hat{k}_{FS} = \max \Bigl\{k \in \{1, \ldots, n : -\frac{1}{k} \sum_{i=1}^k \log(1-p_i) \leq \alpha\}\Bigr\}.
\end{equation}
Applying the FS method with a 5\% significance level across the 85th to 99th quantiles yields a threshold of 16.55, compared to a threshold of 10 obtained using our approach. We select two approaches from the R package: the $DK$ method \citep{Drees1998} and the $hall$ method \citep{Hall1990}. The $DK$ method employs a sequential estimator to determine the optimal sample fraction of the largest order statistics, resulting in a threshold of 14.4. In contrast, the $hall$ method uses a bootstrap evaluation of the AMSE criterion of the Hill estimator to identify this fraction, yielding a value of 12.1. The bulk and tail KL divergence for the test set is reported for each method, using our identified threshold as the separation point.

Similarly to the simulation study, the model using the optimal threshold provides the best fit for the tail, while model combination better captures the bulk of the distribution. Notably, model combination outperforms all methods other than the optimal threshold for both the bulk and the tail. This highlights that while Corollary \ref{cor:GLUE} effectively identifies the optimal threshold, the combination approach mitigates sensitivity to this choice. However, it is important to acknowledge significant uncertainty in the tail estimates due to the limited data and the empirical nature of the density comparisons.

\begin{table}[H]
    \centering
    \begin{tabular}{ccccc}
    \toprule
    & & & \multicolumn{2}{c}{KL divergence}\\
    Method & Threshold & Quantile & Bulk ($x\le 10$) & Tail ($x>10$)\\
    \midrule
    hom-EI & \multicolumn{2}{c}{Combination} & 0.17 & 1.38\\
    w-hom-EI & \multicolumn{2}{c}{Combination} & 0.17 & 1.38\\
    MRL & 10 & 0.95 & 0.19 & 1.33\\
    FS & 16.55 & 0.98 & 0.20 & 2.14\\
    DK & 14.4 & 0.99 & 0.20 & 2.10\\
    hall & 12.1 & 0.96 & 0.19 & 1.80\\
    \bottomrule
    \end{tabular}
    \caption{KL divergence depending on each method's identified threshold for the bulk ($x\le 10$) and tail ($x>10$) of the Danish test set}
    \label{tab:thresholdsDanish}
\end{table}

\subsubsection{Heterogeneous database}

Next, we work with an automobile claims dataset from a Canadian insurer. We have data from over one million claims from 2015 to 2021 for multiple coverages. We choose to study only Vehicle Damage claims in Ontario, which has claims between 2 and 561,000 dollars. Note that although this coverage is capped by the vehicle’s value, we do not select Accident Benefits or Bodily Injury, as these coverages involve preset payment limits, resulting in point masses that complicate threshold identification. In contrast, Vehicle Damage follows a more continuous distribution, avoiding the need for specialised handling of point masses.

We separate data into a training and testing set by taking historical data from 2015 to 2019 as the training data and the more recent 2020 and 2021 losses as the testing data. The training data is further separated to include a calibration component from which combination weights can be calculated by randomly sampling 30\% of claims in the training set.

We set potential thresholds as the 50th to 95th quantiles by jumps of 5\%. We choose these jumps rather than the standard 2.5\% in automatic threshold selection methods to reduce computation time. To use model combination, for each model with a different threshold, we fit a GAMLSS lognormal model on the bulk and a GAM version of the GPD on the tail. For ease of interpretability, Figure \ref{fig:thresholdtype} presents the Dirichlet results of the error integration BMA algorithm, comparing weight variation for two instances of a categorical variable, which we will call "A" and "B", with significantly different distributions, as seen from their quantiles in Table \ref{tab:quartile} and MRL plots in Figure \ref{fig:MRLCoop}. 
\begin{table}[H]
    \centering
    \begin{tabular}{ccc}
    \hline
     & A & B\\
     \hline
    25\% & 217 & 750\\
    50\% & 1485 & 2666\\
    75\% & 4779 & 5994\\
    Variance & $1.3*10^8$ & $7.5*10^7$\\
    \hline
    \end{tabular}
    \caption{Quartile and variance values by category}
    \label{tab:quartile}
\end{table}

\begin{figure}[H]
    \centering
    \includegraphics[width=0.7\linewidth]{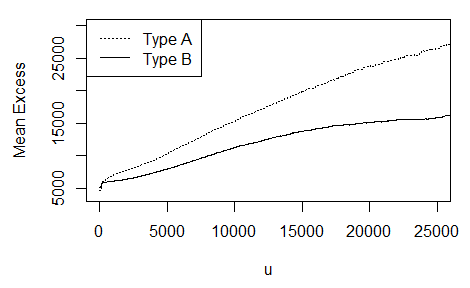}
    \caption{MRL plot by category}
    \label{fig:MRLCoop}
\end{figure}

The threshold identified by our method varies based on the predictive variable, showing a higher threshold for category "B" than for category "A".

\begin{figure}[H]
    \centering
    \begin{minipage}{0.49\linewidth}
        \includegraphics[width=\linewidth]{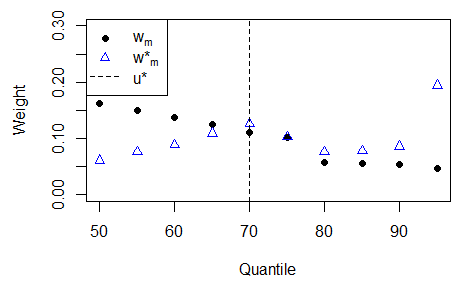}
    \end{minipage}
    \begin{minipage}{0.49\linewidth}
        \includegraphics[width=\linewidth]{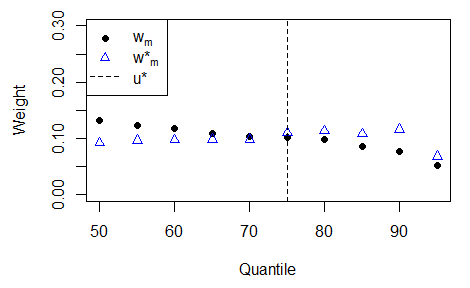}
    \end{minipage}
    \caption{Weights by threshold quantile for categories "A" (left) and "B" (right)}
    \label{fig:thresholdtype}
\end{figure}

Finally, we can once again look at distance and divergence measures to evaluate the combination on a test set compared to the identified threshold. We see in Table \ref{tab:Hellinger_Coop} that model combination outperforms the identified threshold in both the bulk and the tail. Although the tail result seems to contradict the previous simulation study, this can be explained by additional explanatory variables allowing EI to reduce bias in the tail, as well as significant uncertainty due to evaluating with an empirical density with limited data.

\begin{table}[H]
    \centering
    \begin{tabular}{cccccc}
    \hline
    & & \multicolumn{2}{c}{Bulk ($x \le u^{(k)}$)} & \multicolumn{2}{c}{Tail ($x > u^{(k)}$)}\\
    & Method & Hellinger & KL & Hellinger & KL\\
     \hline
    \multirow{3}{*}{\centering Type A} & het-EI & 5.96 & 5.59 & 0.83 & 2.07\\
    & w-het-EI & 5.97 & 5.73 & 0.53 & 1.82\\
    & Identified threshold & 6.21 & 5.93 & 0.82 & 2.96\\
    \hline
    \multirow{3}{*}{\centering Type B} & het-EI & 6.71 & 6.12 & 0.53 & 1.49\\
    & w-het-EI & 6.48 & 5.98 & 0.41 & 1.41\\
    & Identified threshold & 6.85 & 6.32 & 0.59 & 2.15\\
    \hline
    \end{tabular}
    \caption{Hellinger distance $(\times 10^{-4})$ and KL divergence $(\times 10^{-2})$ by combination method}
    \label{tab:Hellinger_Coop}
\end{table}

\section{Conclusion}\label{sec:Conclusion}

In this paper, we proposed modifications to a BMA algorithm, namely the error integration approach proposed by \cite{Jessup2025}. These modifications led to Proposition \ref{prop:Dirichletthreshold} and Corollary \ref{cor:GLUE}, which allowed for identifying the optimal threshold minimising the bias-variance tradeoff in the GPD tail of a distribution. Proposition \ref{prop:Dirichletthreshold} further allowed for these values to depend on predictive variables.

The efficiency of this threshold identification was demonstrated using simulation studies. It was further shown that even with the optimal threshold, when using mixture models, combinations outperformed single threshold models for the bulk of the distribution, and performed better in the tail than models with incorrectly specified thresholds. This was further highlighted using real data with the Danish reinsurance dataset as well as a vehicale damage dataset. These results suggest that while Proposition \ref{prop:Dirichletthreshold} and Corollary \ref{cor:GLUE} identify the optimal threshold, model combination reduces sensitivity to threshold selection.

For future work, it would be interesting to consider multiple coverages simultaneously, such as different car insurance coverages known to be correlated. Dependence could affect threshold values, requiring multiple coverages to be considered together. While this does not affect marginal distributions, it could improve overall loss projections. Additionally, applying our flexible threshold approach to datasets like extreme precipitation, where different stations are treated as variables rather than separately, could increase data availability and enhance analysis.

\newpage
\bibliography{FullListReference}
\bibliographystyle{apalike}

\newpage
\section*{Appendix I - The skew-normal distribution}\label{sec:skewnorm}

The skew-normal distribution, first introduced by \cite{Azzalini1985}, is defined as
\begin{equation*}
    f(x) = \frac{2}{\omega}\phi\left(\frac{x-\xi}{\omega}\right)\Phi\left(\alpha\frac{x-\xi}{\omega}\right),
\end{equation*}
where $\xi$, $\omega$, and $\alpha$ are respectively the location, scale, and shape parameters. $\phi$ and $\Phi$ are the standard normal density and cumulative distribution functions. The relationship between $\xi$ and $\omega$ and the usual normal parameters $\mu$ and $\sigma$ can be expressed through the mean and variance, such that
\begin{align*}
    \mu &= \xi + \omega \frac{\alpha}{\sqrt{1+\alpha^2}}\sqrt{\frac{2}{\pi}}\\
    \sigma^2 &= \omega^2\left(1-\frac{2}{\pi}\frac{\alpha^2}{1+\alpha^2}\right).
\end{align*}
In particular, if $\alpha=0$, then $\xi=\mu$ and $\omega=\sigma$.

The equivalence with the normal distribution can easily be shown by setting $\alpha=0$ as follows.
\begin{align*}
    f(x) &= \frac{2}{\omega}\phi\left(\frac{x-\xi}{\omega}\right)\Phi\left(\alpha\frac{x-\xi}{\omega}\right)\\
    &= \frac{1}{\omega}\phi\left(\frac{x-\xi}{\omega}\right)\\
    &= \frac{1}{\sigma}\phi\left(\frac{x-\mu}{\sigma}\right)\\
    &= \frac{1}{\sqrt{2\pi}\sigma} \exp\left(-\frac{(x-\mu)^2}{2\sigma^2}\right)
\end{align*}

\newpage
\section*{Appendix II - Skew-normal Generalised Likelihood Uncertainty Estimation algorithm}\label{sec:SNGLUE}

\hyperref[tab:algoGLUE]{Algorithm 1} describes a skew-normal approach to GLUE under BMA, where $y^{(q)}$ is the observed $q^\text{th}$ quantile, $y_b^{(q)}$ is the $q^\text{th}$ quantile of the $b^\text{th}$ bootstrap resampling of data $\mathcal{D}$ and $\hat{y}_{m,b}^{(q)}$ is a similar quantile for the $m^\text{th}$ model, with $B$ bootstrap iterations and $Q$ quantiles.
\vspace{-2em}
\begin{center}
    \captionof*{table}{\label{tab:algoGLUE}}
    \resizebox{0.85\textwidth}{!}{
\begin{tabularx}{\textwidth}{l}
        \toprule
        \textbf{Algorithm 1}: Skewed Generalised Likelihood Uncertainty Estimation\\
	\midrule
	1: Resample $\mathcal{D}$ to obtain $B$ bootstrap iterations $y_b^{(q)}$ of the $q^\text{th}$ quantile.\\
	2: Calculate the variance for quantile $q$ as $\sigma^2_q = \frac{1}{B-1} \sum_{b=1}^{B} \left(y_b^{(q)}-\frac{1}{B}\sum_{i=1}^{B}y_i^{(q)}\right)^2$.\\
    3: Calculate the skewness for quantile $q$ as $\gamma_q = \frac{1}{B}\frac{\sum_{b=1}^B \left(y_b^{(q)}-\frac{1}{B}\sum_{i=1}^{B}y_i^{(q)}\right)^3}{\sigma^3_q}$.\\
    4: Calculate the skew-normal parameters as:\\
        \hspace{2cm}$\begin{aligned}
            \delta_q &= \sqrt{\frac{\pi |\gamma_q^{1.5}|}{2(|\gamma_q|^{1.5}+((4-\pi)/2)^{2/3})}}\\
		\alpha_q &= \frac{\delta_q}{\sqrt{1-\delta_q^2}}\\
		\omega_q &= \sqrt{\frac{\pi*\sigma^2_q}{\pi-2\alpha_q^2/(1+\alpha_q^2)}}\\
		\xi_q &= y^{(q)} - \omega_q\sqrt{\frac{2\alpha_q}{\pi(1+\alpha_q^2)}}
        \end{aligned}$\\
    5: Calculate the likelihood and weighted-likelihood assuming residuals follow a skew-\\
    normal distribution, with $\phi$ and $\Phi$ the standard normal density and cumulative\\ 
    function respectively:\\
        \begin{minipage}{0.96\linewidth}
            \begin{equation*}
            L(\hat{y}_m^{(q)}) = \frac{2}{\omega_q} \left(\prod_{b=1}^B \phi\left(\frac{\hat{y}_{m,b}^{(q)}-\xi_q}{\omega_q}\right)\Phi\left(\alpha_q \left(\frac{\hat{y}_{m,b}^{(q)}-\xi_q}{\omega_q}\right)\right)\right)^{1/B}
            \end{equation*}
            \hspace{2cm}$\begin{aligned}
                L(\mathcal{D}|\mathcal{M}_m) &= \frac{1}{Q} \sum_{q=1}^{Q} L(\hat{y}_m^{(q)}).\\
            L^*(\mathcal{D}|\mathcal{M}_m) &= \sum_{q=1}^{Q} \frac{\hat{y}_m^{(q)}}{\sum_{j=1}^Q \hat{y}_m^{(j)}} L(\hat{y}_m^{(q)})
            \end{aligned}$
        \end{minipage}\\
	6: Update the probability of each model as\\
    \hspace{2cm}$\begin{aligned}
        w_m &= \frac{L(\mathcal{D}|\mathcal{M}_m) \Pr(\mathcal{M}_m)}{\sum_{l=1}^{M}L(\mathcal{D}|\mathcal{M}_l)\Pr(\mathcal{M}_l)},\\
        w^*_m &= \frac{L^*(\mathcal{D}|\mathcal{M}_m) \Pr(\mathcal{M}_m)}{\sum_{l=1}^{M}L^*(\mathcal{D}|\mathcal{M}_l)\Pr(\mathcal{M}_l)}.
    \end{aligned}$\\
    7: Identify the optimal threshold $u^*$ with corresponding model $\mathcal{M}_{m^*}$ as\\
        \hspace{2cm}$\begin{aligned}
            u^* = \argmin_m(u_m : w^*_m \ge w_m).
        \end{aligned}$\\
        8: Calculate posterior distributions as\\
        \hspace{2cm}$\begin{aligned}
            \Pr(y|\mathcal{D}) &= \sum_{m=1}^M \Pr(y|\mathcal{M}_m)\Pr(\mathcal{M}_m|\mathcal{D}),\\
                \Pr\,_{u^*}(y|\mathcal{D}) &= \Pr(y|\mathcal{M}_{m^*}).
        \end{aligned}$\\
	\bottomrule
    \end{tabularx}
}
\end{center}

\section*{Appendix III - Distance and divergence measures}\label{sec:measures}

Hellinger distance is defined (under Lebesgue measure) as
\begin{align*}
    H^2(f,g) &= \frac{1}{2} \int \left(\sqrt{f(x)}-\sqrt{g(x)}\right)^2 dx,\\
    \intertext{which can easily be shown to be equivalent to}
    &= 1-\int \sqrt{f(x)g(x)} dx, \nonumber
\end{align*}
where $H^2(f,g)=0$ if $f=g$ and $H^2(f,g)=1$ is the case where $f$ and $g$ have entirely different supports.

KL divergence between distributions $\mathbb{P}$ and $\mathbb{Q}$ is defined as
\begin{equation*}
	D_{KL}(\mathbb{P} || \mathbb{Q}) = \int f_{\mathbb{P}}(x) \log\left(\frac{f_{\mathbb{P}}(x)}{f_{\mathbb{Q}}(x)}\right) dx.
\end{equation*}

\end{document}